\documentclass[epj]{svjour}
\usepackage{latexsym}
\usepackage{graphics}
\usepackage{amsmath,amssymb}
\usepackage{physics}
\usepackage{bm}
\DeclareGraphicsExtensions{.pdf,.png,.jpg}

\begin{document}
\title{A new phenomenological definition of entropy and application to black holes}
\author{Taha A Malik\inst{1} \email{taha.malik@utsa.edu, taha.malik11@alumni.imperial.ac.uk}\and  Rafael Lopez-Mobilia\inst{1}
\email{rafael.lopezmobilia@utsa.edu}
}                     
%
%
\institute{Department of Physics and Astronomy, The University of Texas at San Antonio 
}
\date{Received: date / Revised version: date}
%
\abstract{
Typically, the entropy of an isolated system in equilibrium is calculated by counting the number of accessible microstates, or in more general cases by using the Gibbs formula. In irreversible processes entropy spontaneously increases and this is understood from statistical arguments. We propose a new measure of entropy based on the level of irreversibility of a process. This formulation agrees in first approximation with the usual methods of calculating entropy and can be readily applied in the case of a black hole in the semiclassical regime.
\PACS{
      {PACS-key}{discribing text of that key}   \and
      {PACS-key}{discribing text of that key}
     } 
} 
\authorrunning{T. A. Malik and R.Lopez-Mobilia }
\titlerunning{Time relative entropy}
\maketitle
\section{Introduction}

Classically, the laws of black hole mechanics closely resemble the laws of thermodynamics \cite{10.2307/79210,10.1007/BF01645742}. In particular, the surface area of a black hole is always non-decreasing, analogous to the 2nd law of thermodynamics which states that entropy is always non decreasing. When quantum effects are included, the laws of black hole mechanics are interpreted as being true thermodynamics properties \cite{PhysRevLett.26.1344}. A black hole is found to have entropy $S=\frac{A}{4}$ in Planck units where $A$ is the surface area of the black hole \cite{PhysRevD.7.2333}. After this identification, there is mounting evidence that the generalized second law of thermodynamics holds at least semi-classically \cite{PhysRevD.85.104049,PhysRevD.82.124019}.

However, there is no widely accepted interpretation of what or even where the degrees of freedom of a black hole's entropy are. Perhaps a quantum theory of gravity will be required to resolve this question and a deeper understanding of the notion of entropy itself. Motivated by this, we take the second law of thermodynamics to be the most fundamental property of entropy and use it as a guide to construct a new measure of entropy, the time relative entropy.   

The time relative entropy is a measure of the irreversibility of a process relative to another process and, to first approximation, we show that the time relative entropy agrees with the usual methods of calculating entropy and can be readily applied to the case of a black hole in the semiclassical regime. We hope that the time relative entropy may lead to a better understanding of some of the above issues. Perhaps more fundamentally, entropy has more do to with irreversibility and the arrow of time than the microstates themselves.

\section{Semi-classical dynamics}

We describe a dynamical system with a state $p$ being a vector describing  a probability distribution over a set $A$ where $A$ is interpreted as a set of accessible cells of the system. 
For a fixed time step $\delta t$, we describe the dynamics of the state with a probability transition matrix (PTM), $D(\delta t)$ . In this way we model the dynamics of a system using a Markov chain so $p(k\delta t)=D^k(\delta t)p(0)$. The above can be viewed as a discretized version of dynamics on some phase space.  

For example, in this paper, we will be considering a semi-classical system. To obtain our semi-classical system using the model above, we start with a classical system with $N$ particles defined by a time independent Hamiltonian on some phase space at fixed energy $E$. We set $A(E)$ to be the accessible phase space of the system at fixed energy $E$ partitioned by cells of size $(\hbar/2)^{3N}$. Semi-classically, these cells represent with maximum certainty the momenta and positions of the $N$ particles. We will assume that $A$ is finite. Since we can not precisely calculate trajectories of these $N$ particles due to quantum fluctuations, we use a probability distribution $p$ over $A$ rather than a point in phase space and describe the state with dynamics described with a PTM. We describe how to construct the PTM in section \ref{003.SD}.
\subsection{Time relative entropy}
With the above system, define $K$ to be a subsystem if $K \subset A$ and $K \neq \emptyset$

Let $K_x$ and $K_y$ be two subsystems such that $K_x \cap K_y = \emptyset$. we define the time relative entropy of $K_x$ with respect to $K_y$ as 
\begin{equation}\label{003.RE}
S(K_x|K_y) = \text{log}\left(\frac{\tau(K_x \rightarrow K_y)}{\tau(K_y \rightarrow K_x)}\right),
\end{equation}
where $\tau(K_x \rightarrow K_y)$ is the expected time for the system initially in a cell in $K_x$ to evolve into a cell in $K_y$, averaged over all the cells in $K_x$. 

Note that if $S(K_x|K_y)>0 $, then it takes longer for the system to go from a cell $K_x$ to $K_y$ than it does from a cell in $K_y$ to $K_x$ on average. In this case we say that the process of evolving from $K_x$ to $K_y$ is more irreversible than the process of evolving from $K_y$ to $K_x$.


\section{Ideal gas heuristic example}

Consider the microcanonical ensemble of an ideal gas in a box with volume $V$ and $N$ particles at fixed energy $E$. Let $A$ be the accessible phase space partitioned into cells of size $(\hbar/2)^{3N}$. From the uncertainty principle, these cells represents the possible points the system could be in phase space with maximum precision semi-classically.  
\begin{itemize}
	\item Let $K_y$ be the subsystem corresponding to the collection of cells where the $N$ particles are all on one side of the box, contained within volume $V'<V$. 
	\item Let $K_x=A \backslash K_y$, which is approximately the whole phase space provided $|K_x|>>|K_y|$. This will be the case for `generic' choices of $V'$. 
\end{itemize}
Since below we give only a heuristic calculation for a time relative entropy as an example, `generic' and be though of as a reasonable choice.

As a rough estimate for the ratio of the expected times, we discretize time by the characteristic time step $\Delta t$ defined as the minimum time required for a particle to move from one side of the box to the other at its average velocity, so
\begin{equation}
\Delta t \propto \frac{V^\frac{1}{3}}{<velocity>}.
\end{equation}
After this time, we expect the system to be approximately in any cell in $A$ with equal probability since the particles have had enough time to transverse anywhere within the box. We make the following estimations:
\begin{equation}\label{003.T1}
\tau(K_y \rightarrow K_x) \approx  \Delta t
\end{equation}
since $|K_x|$ is much larger than $|K_y|$ and so similar to the diffusion of a gas, we expect this process to be quick and take only one time step. Another way to see this is to view $\tau(K_y \rightarrow K_x)$ as the expected time it takes at least one particle to leave the volume $V'$.
\begin{equation}\label{003.T2}
\tau(K_x \rightarrow K_y) \approx \Delta t\frac{|K_x|}{|K_y|}=\Delta t \left(\frac{V}{V'}\right)^N
\end{equation}
since after $|A|$ time steps, we expect the to system transverses every cell once. In this time, the system is in $K_x$ for $|K_x|$ time steps and $K_y$ for $|K_y|$ time steps. Hence 
\begin{equation}\label{003.TRE}
S(K_x|K_y) = log\left(\left(\frac{V}{V'}\right)^N\right)
\end{equation}
It agrees with the usual formula for the differences of entropies of two microcanonical ensembles with different volumes\footnote{We set Boltzmann constant equal to 1 ($k=1$)}, so Eq. (\ref{003.TRE}) can be interpreted a relative entropy. Explicitly
\begin{equation}
S(K_x|K_y) \propto (S(E,V,N)-S(E,V',N)),
\end{equation}
where $S(E,V,N)$ is the entropy of an ideal gas in the microcanonical ensemble. At least at this heuristic level, the time relative entropy agrees with the usual formula for entropy. This correspondence extends to the canonical ensemble since the canonical ensemble can be obtained by placing the system in a large heat bath, with the heat bath modeled using the microcanonical ensemble. 

\begin{remark}
One may wonder why the logic that applied Eq. to (\ref{003.T2}) can not be applied to Eq. (\ref{003.T1}). Notice that if we applied the same logic, then we would have a time step which is smaller than the discretized time step. Below we outline how to handle more general time steps.
\end{remark}

\section{Dynamics}\label{003.SD}

Without loss of generality, given a semi-classical system as described above, let the cells in $A$ be labelled so that  $A=\{c_1,...,c_\Omega\}$ where $\Omega=|A|$.
Then $p=p(t_0)$ becomes a $\Omega$-dimensional vector such that each entry, $i$, is the probability that the system is in cell $C=c_i$ at some time $t_0$ ($p(t_0)_i=P(C=c_i)$). After an arbitrary small choice for the time step, $\delta t$, we define the probability transition matrix (PTM) $D=D(\delta t)$ as 
\begin{equation}
D(\delta t)_{ij} = P(C=c_i,t_0+\delta t|C=c_j, t_0)
\end{equation} 
and after applying the law of total probabilities we have
\begin{equation}
D(\delta t)p(t_0)=p(t_0+\delta t),
\end{equation}
\begin{equation}
D(\delta t)^kp(t_0)=p(t_0+k\delta t),
\end{equation}
where we have assumed that the state $p(t_0+\delta t)$ only depends on the state $p(t_0)$ and so the dynamics can be modeled using a Markov chain\footnote{If $H$ does not explicitly depend on time, then $D$ also does not explicitly depend on time. And since the flow after time $\delta t$ only depends on the initial conditions $z_0$ and not on the entire history of the flow, we find that the dynamics can be modeled as a Markov chain.}.

Using the classical system from which we arrive at our semi-classical system, we can define $P(C=c_i,t_0+\delta t|C=c_j, t_0)$ by calculating how much of the flow, generated by classical Hamiltonian, from cell $c_j$ at time $t_0$ enters the cell $c_i$ after time $\delta t$ in phase space (Fig. (\ref{003.fig:H})). 

\begin{figure}\centering
  \resizebox{0.40\textwidth}{!}{\includegraphics{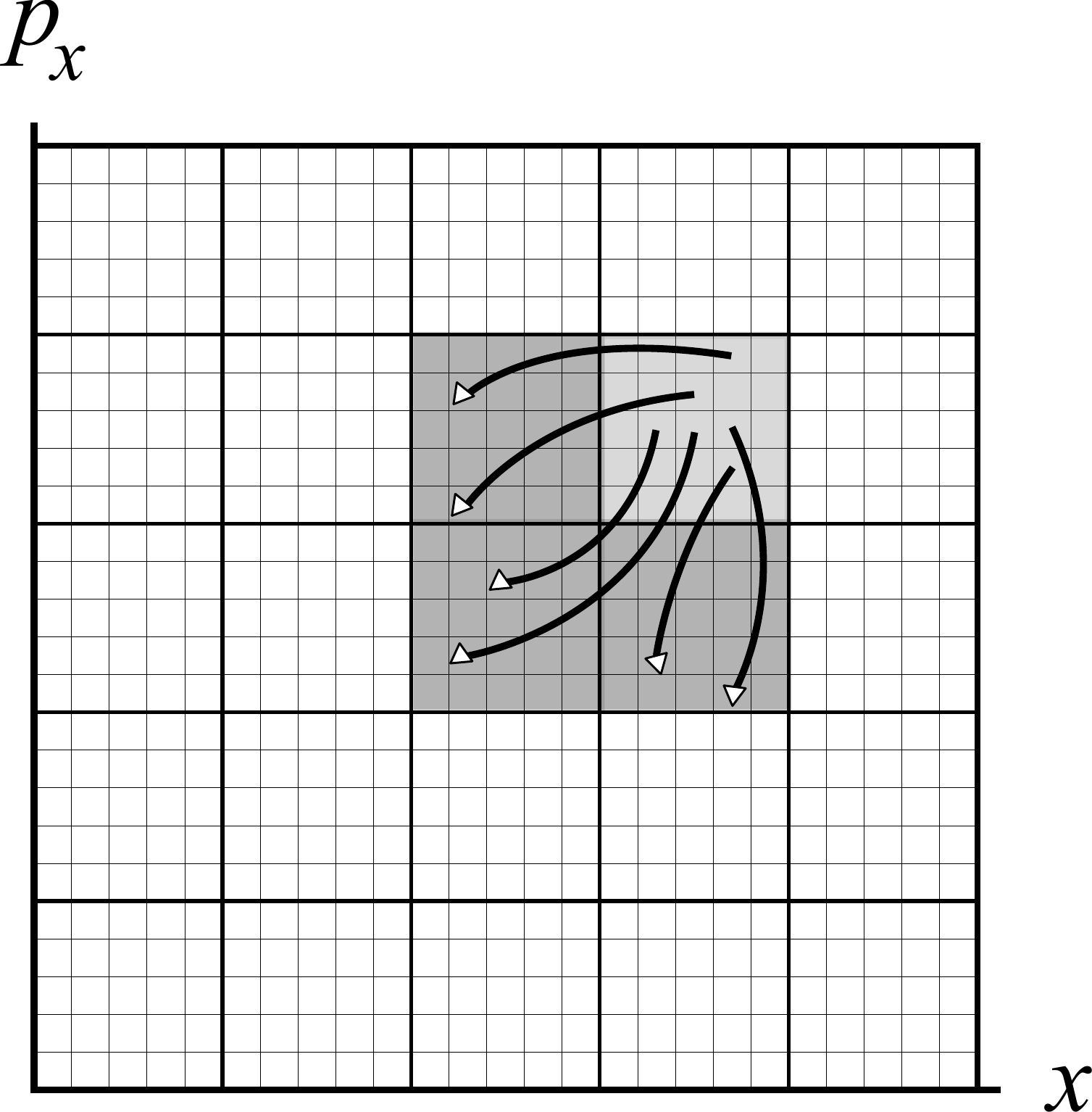}}
  \caption{Flow generated by the Hamiltionion on a simplified phase space}
  \label{003.fig:H}
\end{figure}

\begin{definition}
	Given a classical system defined by a Hamiltonian on a phase space and after fixing a time step $\delta t$, define $D(\delta t)$ via $P(C=c_i,t_0+\delta t|C=c_j, t_0)$ as
	\begin{align}
		D(\delta t)_{ij}=P(C=c_i,t_0+\delta & t|C=c_j, t_0)=\int_{c{_j}}\rho_i\circ\gamma_{\delta t}(z)dz \nonumber \\
		&=\frac{1}{(\hbar/2)^{3N}}\int_{c{_j}\cap\gamma_{\delta t}(c{_i})}dz,
	\end{align}
	where 
	$\rho_i(z)=\{\frac{1}{(\hbar/2)^{3N}} : z\in c_i , 0 : otherwise\}$, 
	$\gamma_t(z_0)$ is the flow generated in phase space by $H$ with initial condition $z_0=(x_0,p_0)$ at time $t=t_0$ and
	$dz=dx^ndp^n$ is the volume element in phase space $P$.
	
	\end{definition} 
This procedure allows us to take a classical system and define its dynamics semi-classically. 

\begin{theorem} 
$D(\delta t)$ converges  	
\end{theorem}

\begin{proof}
	See known Markov chain result in the literature\cite{J.R.Norris:1998aa}. 
\end{proof}

\begin{remark} ~\

$\gamma_{-t}\circ\gamma_t(z_0)=z_0$

$\gamma_t=P\circ\gamma_{-t}\circ P$ such that $P(x,p)=(x,-p)$

\end{remark}

\begin{proposition}
	$D(\delta t)_{ij}=D(\delta t)_{j'i'}$, where $i'$ is the label for cell $P(c_i):=c_{i'}$.
\end{proposition}

\begin{proof}
~\\

We show that $\int_{c{_j}}\rho_i\circ\gamma_{\delta t}(z)dz=\int_{c{_{i'}}}\rho_{j'}\circ\gamma_{\delta t}(z)dz$ \\

$\int_{c{_j}}\rho_i\circ\gamma_{\delta t}(z)dz=\frac{1}{(\hbar/2)^{3N}}\int_{c{_j}\cap\gamma_{\delta t}(c{_i})}dz$\\

$=\frac{1}{(\hbar/2)^{3N}}\int_{\gamma_{-\delta t}(c{_j})\cap c{_i}}dz$

$=\frac{1}{(\hbar/2)^{3N}}\int_{P\circ \gamma_{\delta t}(P(c{_j}))\cap c{_i}}dz$

$=\frac{1}{(\hbar/2)^{3N}}\int_{\gamma_{\delta t}(P(c{_j}))\cap P(c{_i})}dz=\int_{c{_{i'}}}\rho_{j'}\circ\gamma_{\delta t}(z)dz$.\\
where we use that $\gamma_t$ and $P$ are volume preserving maps in phase space. 

\end{proof}

\begin{corollary}\label{003.v}
	$\boldsymbol{1}=[1,1,...,1]^T$ is an eigenvector of $D=D(\delta t)$ with eigenvalue 1
\end{corollary}

\begin{proof}
Since $D$ is a PTM, $\sum_{ij}D_{ij}p_j=1$ for all probability distributions $p_j$ (Sums of probabilities must equal 1 and both $p_j$ and $\sum_{ij}D_{ij}p_j$ are probability distributions). In particular, for $p_j=\delta_{jk}$

$1=\sum_{ij}D_{ij}p_j=\sum_iD_{ik}$ so the sum of any column of $D$ equals 1. Therefore, $1=\sum_iD_{ik}=\sum_iD_{k'i'}=\sum_{i'}D_{k'i'}=\sum_iD_{k'i}$ so sum of any row of $D$ equals 1.\\

Hence $\sum_jD_{ij}\boldsymbol{1}_j=\sum_jD_{ij}=1$ or $D\boldsymbol{1}=\boldsymbol{1}$. \\

Similarly, $\boldsymbol{1}^TD=\boldsymbol{1}^T$.

\end{proof}

We will assume that the system is ergodic so the eigenvalue 1 is not degenerate.

\begin{corollary}
	$Lim_{r\longrightarrow \infty}(D(\delta t)^r)=\frac{1}{\Omega}M$, where $M$ is a matrix such that all its entries are 1
\end{corollary}

\begin{proof}
See known Markov chain result in the literature\cite{J.R.Norris:1998aa}. 
\end{proof}
 
\section{Entropy of microcanonical ensemble}

Suppose we have a Hamiltonian system with $N$ particles with fixed total energy $E$, which we wish to model semi-classically. Let $K_1 \subset A$ be the subsystem of one (arbitrary) cell and let $K_{all} \subset A=A \backslash K_1$ be the subsystem of all the other cells. Without loss of generality label the arbitrary cell as the last label. Then 
\begin{equation}
S(K_{all}|K_1) = \text{log}\left(\frac{\tau(K_{all} \rightarrow K_1)}{\tau(K_1 \rightarrow K_{all})}\right),
\end{equation}
where

$\tau(K_{all} \rightarrow K_1)=$

$\delta t[v_1^TD(\delta t)p_{all}+2(v_{all}^TD(\delta t)p_{all})(v_1^TD(\delta t)^2p_{all})$

$+3(v_{all}^TD(\delta t)p_{all})(v_{all}^TD(\delta t)^2p_{all})(v_1^TD(\delta t)^3p_{all})$

$+...] $

$$
=\delta t \sum_{n=1}^{\infty}n\Pi_{k=0}^{n-1}\left(v_{all}^TD(\delta t)^kp_{all}\right)(v_1^TD(\delta t)^np_{all})
$$
with

$p_{all}^T=\frac{1}{\Omega-1}[1,1,1,...,0]$, $p_1^T=[0,0,0,...,1]$, 

$v_{all}^T=[1,1,1,...,0]$, $v_1^T=[0,0,0,...,1]$.

A similar expression can be given for $\tau(K_{1} \rightarrow K_{all})$ by swapping the (all) and (1) index\footnote{In words, 

$\tau(K_{all} \rightarrow K_1)=$

$\delta t \sum_{n=1}^{\infty}n\Pi_{k=0}^{n-1}((\text{Prob of system in }K_{all}\text{ after k steps})$

$\times(\text{Prob of system in }K_{1}\text{ after n steps}))$}
. Hence
\begin{align}
&S(K_{all}|K_1)=\nonumber \\
&\text{log}\left(
\frac{\delta t \sum_{n=1}^{\infty}n\Pi_{k=0}^{n-1}\left(v_{all}^TD(\delta t)^kp_{all}\right)(v_1^TD(\delta t)^np_{all})}{\delta t \sum_{n=1}^{\infty}n\Pi_{k=0}^{n-1}\left(v_{1}^TD(\delta t)^kp_{1}\right)(v_{all}^TD(\delta t)^np_{1})}\right).
\end{align}
\begin{remark}
 $(v_i^TD(\delta t)^rp_j)$ is the probability of being in a cell in $K_i$ after $r$ steps with the state initially uniformly distributed on $K_j$. 
\end{remark}

\section{Calculation of $S(K_{all}|K_1)$} 

Note that any term of the form $v^T_iD^kp_j$ can be rewritten in terms of $v^T_1D^kp_1$ using corollary (\ref{003.v}) so the entropy above can be written as \footnote{Example: $v^T_{all}D^kp_1=(\boldsymbol{1}^T-v_1^T)D^kp_1=1-v_1^TD^kp_1$}
\begin{align}\label{003.RE2}
& S(K_{all}|K_1)= \nonumber \\
&\text{log}\left(
\frac{\sum_{n=1}^{\infty}n\Pi_{k=0}^{n-1}\left[1-\frac{1}{\Omega -1}(1-f(k))\right](1-f(n))\frac{1}{\Omega-1}}{\sum_{n=1}^{\infty}n\Pi_{k=0}^{n-1}\left[f(k)\right](1-f(n))}\right),
\end{align}
where $f(k)=v_1D^k(\delta t)p_1$, which is the probability that the system remains in cell $K_1$ after $k$ time steps.

\subsection{Test functions for $f$}

To understand Eq. (\ref{003.RE2}), we first calculate the expression using a test function for $f$. We choose $f$ to be
\begin{equation}
f(k)=(1-\frac{1}{\lambda}k) : k \leq k' \\
, f(k)=\frac{1}{\Omega} : k \geq k',
\end{equation}
where $k'$ is defined such that $1-\frac{1}{\lambda}k'=\frac{1}{\Omega}$ so $k' \approx \lambda$. Without loss of generality we will assume that $\lambda$ is an integer. If we assume that  $1<\lambda<<\Omega$ then the numerator in Eq. (\ref{003.RE2}) becomes 
\begin{subequations}

\begin{equation}
\sum_{n=1}^{\infty}n\Pi_{k=0}^{n-1}\left[1-\frac{1}{\Omega -1}(1-f(k))\right](1-f(n))\frac{1}{\Omega-1}
\end{equation}
\begin{align}
&\approx\frac{1}{\Omega}\sum_{n=1}^{k'}n\Pi_{k=0}^{n-1}\left[1-\frac{1}{\Omega}\left(\frac{k}{\lambda}\right)\right]\left(\frac{n}{\lambda}\right)\nonumber \\
&+\frac{1}{\Omega}\sum_{n=k'+1}^{\infty}n\Pi_{k=0}^{k'}\left[1-\frac{1}{\Omega}\left(\frac{k}{\lambda}\right)\right]\Pi_{k=k'}^{n-1}\left[1-\frac{1}{\Omega}\right]
\end{align}
\begin{align}\label{003.16}
&=\frac{1}{\Omega}\sum_{n=1}^{k'}\frac{n^2}{\lambda}\left(\frac{1}{\Omega \lambda}\right)^n \frac{\Gamma(\Omega \lambda+1)}{\Gamma(\Omega \lambda-n+1)} \nonumber \\
&+\frac{1}{\Omega}\sum_{n=k'+1}^{\infty}n\left(\frac{1}{\Omega \lambda}\right)^{k'} \frac{\Gamma(\Omega \lambda+1)}{\Gamma(\Omega \lambda-k'+1)}\left(1-\frac{1}{\Omega}\right)^{n-k'}.
\end{align}	
\end{subequations}
Similarly, the denominator is approximately given by 
\begin{equation}
\sum_{n=1}^{k'}\frac{n^2}{\lambda}\left(\frac{1}{\lambda}\right)^n \frac{\Gamma(\lambda+1)}{\Gamma(\lambda-n+1)} +\frac{1}{\Omega}\sum_{n=k'+1}^{\infty}n\left(\frac{1}{\lambda}\right)^{k'} \lambda ! \left({\frac{1}{\Omega}}\right)^{n-k'}.
\end{equation}
Our assumption implies that $\lambda << \Omega < \Omega \lambda << \Omega^2$ so that $\frac{\Gamma(\Omega \lambda+1)}{\Gamma(\Omega \lambda-z+1)} \approx (\Omega \lambda)^z$ for any $z$ between $1$ and $\lambda$ and Eq. (\ref{003.16}) becomes
\begin{equation}
\frac{1}{\Omega}\sum_{n=1}^{k'}\frac{n^2}{\lambda}+\frac{1}{\Omega}\sum_{n=k'+1}^{\infty}n\left(1-\frac{1}{\Omega}\right)^{n-k'}.
\end{equation}
Hence 
\begin{align}
&S(K_{all}|K_1)\approx \nonumber \\
&\frac{\frac{1}{\Omega}\sum_{n=1}^{k'}\frac{n^2}{\lambda}+\frac{1}{\Omega}\sum_{n=k'+1}^{\infty}n\left(1-\frac{1}{\Omega}\right)^{n-k'}}{\sum_{n=1}^{k'}\frac{n^2}{\lambda}\left(\frac{1}{\lambda}\right)^n \frac{\Gamma(\lambda+1)}{\Gamma(\lambda-n+1)} +\frac{1}{\Omega}\sum_{n=k'+1}^{\infty}n\left(\frac{1}{\lambda}\right)^{k'} \lambda ! \left({\frac{1}{\Omega}}\right)^{n-k'}} \label{003.19} \\
&\approx \frac{\frac{1}{\lambda}[\frac{1}{6}k'(k'+1)(2k'+1)]+(\frac{1}{\Omega}k'+1)\Omega^2}{\Omega h(k') + (1+k'-\frac{1}{\Omega}k')\frac{k'!}{{k'^{k'}}}} \label{003.20} \nonumber \\
& \approx \left(\frac{\frac{1}{\Omega}k'+1}{h(k')}\right)\Omega+\frac{1}{3}\frac{1}{h(k') \Omega}k'^2+O(k'),
\end{align}
where 
\begin{align}\label{003.21}
h(k')=\sum_{n=1}^{k'}\frac{n^2}{\lambda}&\left(\frac{1}{\lambda}\right)^n \frac{\Gamma(\lambda+1)}{\Gamma(\lambda-n+1)} \nonumber \\
& \approx\sum_{n=1}^{k'}\frac{n^2}{k'}\left(\frac{1}{k'}\right)^n \frac{k' !}{(k'-n)!}.
\end{align}
Note that in Eq. (\ref{003.19}), (\ref{003.20}) and (\ref{003.21}), we used that $\lambda \approx k'$

\begin{conjecture}

$\frac{n^2}{k'}\left(\frac{1}{k'}\right)^n\frac{k'!}{(k'-n)!}\leq 1$ for integer $n\in[0,k']$.

\end{conjecture}

\begin{proof}
We have not found a proof for this conjecture. However, we have analyzed it numerically and it seems to hold. We leave the proof of this conjecture for future work.
\end{proof}

\begin{corollary}
$h(k')\in(\frac{1}{k'},k')$.
\end{corollary}

\begin{remark}
Even if the lemma is not true, we find that 

$\frac{n^2}{k'}\left(\frac{1}{k'}\right)^n\frac{k'!}{(k'-n)!}$ to be of order of at most unity and to be less than one for almost all values of $n$. Thus we still expect the corollary to hold.  
\end{remark}

Hence for $1<\lambda <<\Omega$, we obtain using the corollary, 

$$\Omega h(k') >>(1+k')/k' > (1+k'-\frac{1}{\Omega}k')/k'\times\frac{k'!}{k'^{k'-1}}$$ 
so the term $(1+k'-\frac{1}{\Omega}k')\frac{k'!}{k'^{k'}}$ can be ignored in Eq. (\ref{003.20}). Hence
\begin{align}\label{003.A}
&\frac{\frac{1}{\lambda}[\frac{1}{6}k'(k'+1)(2k'+1)]+(\frac{1}{\Omega}k'+1)\Omega^2}{\Omega h(k') + (1+k'-\frac{1}{\Omega}k')\frac{k'!}{{k'^{k'}}}}\nonumber \\
& \approx \left(\frac{\frac{1}{\Omega}k'+1}{h(k')}\right)\Omega+\frac{1}{3}\frac{1}{h(k') \Omega}(k'^2+\ldots),
\end{align}
where $\ldots$ contain lower order terms in $k'$

Now, $\frac{k'^2}{\Omega}<<\frac{\Omega^2}{\Omega}=\Omega$ and so 
\begin{equation}
\left(\frac{\frac{1}{\Omega}k'+1}{h(k')}\right)\Omega  \approx \frac{\Omega}{h(k')} >> \frac{1}{3}\frac{1}{h(k')\Omega} k'^2.	
\end{equation}
 This means the $\frac{1}{3}\frac{1}{h(k')\Omega} k'^2$ term can be ignored in Eq. (\ref{003.A}) and so finally,
\begin{align}
&S(K_{all}|K_1)\approx \text{log}\left(\left(\frac{\frac{1}{\Omega}k'+1}{h(k')}\right)\Omega \right)   \approx \text{log}(\frac{\Omega}{h(k')}) \nonumber \\
&= \text{log}(\Omega)-\text{log}(h(k')) \approx \text{log}(\Omega).
\end{align}

This by using the time relative entropy, we have recovered the standard entropy for the microcanonical ensemble. We will use Eq. (\ref{003.RE2}) to define the total entropy for the microcanonical ensemble.  However, it may seem that the above result holds only for special choices for $f$. We show below that for any choice of $f$ which decays fast enough, the above results still holds. 

\begin{theorem}
If $\Omega$ is very large and $f$ satisfies the following conditions
\begin{enumerate}
\item $f(0)=1$ (Automatically satisfied from the definition of $f$),
\item $lim_{n\rightarrow \infty}f(n) \rightarrow \frac{1}{\Omega}$ (Automatically satisfied from the definition of $f$),
\item $0>f'>-\infty$ ,
\item $\frac{1}{|f'|}<<\Omega$ while $\frac{1}{f}<<\Omega$,
\item $f''>0$ or f' increasing,
\end{enumerate}
then 
\begin{equation}
	S(K_{all}|K_1)\approx \text{log}(\Omega)
\end{equation}
\end{theorem}

\begin{proof}
 Write $f$ as

$f(k)=1-g(k)$ for $1\leq k<k'$  

$f(k)\approx \frac{1}{\Omega}$ for $k>k'$

where $k'$ is defined so that 
\begin{equation}
	1-g(k')\approx \frac{1}{\Omega} \implies g(k')\approx 1.
\end{equation}
Since $\frac{1}{|f'|}<<\Omega$ we have that  $k'<<\Omega$. And with condition 5, we have that 
\begin{equation}
1-g(k) \leq 1-\frac{k}{k'}
\end{equation}
for $k<k'$. The total entropy can be approximately written as 
\begin{equation}
\text{log}\left(\frac{A+B(\frac{1}{\Omega}k'+1)\Omega^2}{\Omega C+D(1+k'-\frac{1}{\Omega}k')}\right),
\end{equation}
where 
\begin{subequations}

\begin{equation}
A=\sum_{n=1}^{k'}n\Pi_{k=0}^{n-1}\left[1-\frac{1}{\Omega}(g(k))\right](g(n)),
\end{equation}
\begin{equation}
B=\Pi_{k=0}^{k'-1}\left[1-\frac{1}{\Omega}(g(k))\right],
\end{equation}
\begin{equation}
C=\sum_{n=1}^{k'}n\Pi_{k=0}^{n-1}\left[f(k)\right]g(n),
\end{equation}
\begin{equation}
D=\Pi_{k=0}^{k'-1}\left[f(k)\right].
\end{equation}
\end{subequations}
Assume for now that $\Omega C>>D(1+k'-\frac{1}{\Omega}k')$. Then the total  entropy can be approximated by 

\begin{equation}
	\text{log}\left(\frac{A+B(\frac{1}{\Omega}k'+1)\Omega^2}{\Omega C}\right).
\end{equation}

We have that 
\begin{align}
C&\leq\sum_{n=1}^{k'}n\Pi_{k=0}^{n-1}\left(1-\frac{k}{k'}\right)=\sum_{n=1}^{k'}n\left(\frac{1}{k'}\right)^n\frac{k'!}{(k'-n)!} \nonumber \\
	&\leq \sum k' \leq k'^2
\end{align}
and 
\begin{align}
	C\geq\sum_{n=1}^{k'}n\Pi_{k=0}^{n-1}\left(\frac{1}{\Omega}\right)\left(\frac{n}{k'}\right)=\sum\frac{n^2}{k'}\left(\frac{1}{\Omega}\right)^{n-1}= O \frac{1}{k'},
\end{align}
where $O\approx 1$. Similarly, find that have that $A\leq \frac{1}{2}(k'^2+k')$ and $B\approx 1$. Hence
\begin{equation}
S=\text{log}\left(\frac{A}{\Omega C}+B\frac{\frac{1}{\Omega}k' +1}{C}\Omega \right) \approx \text{log}(\Omega).
\end{equation}
Finally, to show that $\Omega C>>D(1+k'-\frac{1}{\Omega}k')$, we have that 
\begin{equation}
	D(1+k'-\frac{1}{\Omega}k')\leq k'!\left(\frac{1}{k'}\right)^{k'}(1+k')\leq 2 << C\Omega.
\end{equation}
\end{proof}
\begin{remark}
This proof should not be considered as completely rigorous. In particular, $k'$ may not be so sharply defined. We leave a rigorous proof of this for future work. 
\end{remark}

One may criticize that since $f$ depends on the choice of  $\delta t$, by making $\delta t$ small enough one can also make $f$ decay arbitrarily slowly and break condition 4. However, if $\delta t$ is chosen too small, then the semi-classical approximation breaks down as only the flow at the boundary of the cells (Fig. (\ref{003.fig:H})) contributes to $D$.

For the semi-classical approximation to remain valid, one must allow the flow to `well mix' within a cell, which we will take to be the minimum time required for the flow to transverse from one side of the cell to the other side. $\delta t_{LB} = \frac{\Delta x}{v}$ gives a lower bound for $\delta t$, where $\Delta x$ is the spacial dimensions/size of the cell and $v$ is the average speed of the particles.

By the uncertainty relation, we have $\Delta x \Delta p\geq \hbar/2$. Using the relation $P^2=2mE$ for non relativistic system, we find that
\begin{equation}
\delta t_{LB}=\frac{\hbar/2}{\Delta E}>>\frac{\hbar}{E}
\end{equation}
which is simply the time energy uncertainly relation. For a choice of $\delta t \geq \delta t_{LB}$, we expect that $f$ decays fast enough to satisfy the conditions of the theorem. It would be interesting to further investigate this with computer simulations and other analysis. 

\section{Canonical and Grand Canonical ensemble}
\label{003.CE}

The usual expression for entropy can also be recovered for the canonical and grand canonical ensemble via the microcanonical ensemble by putting the system in a large heat bath. However, one needs to be careful since single celled subsystems can have different time relative entropies with respect to each other, and there is no natural choice for picking out a such a preferred subsystem.

\begin{figure}\centering
  \resizebox{0.40\textwidth}{!}{\includegraphics{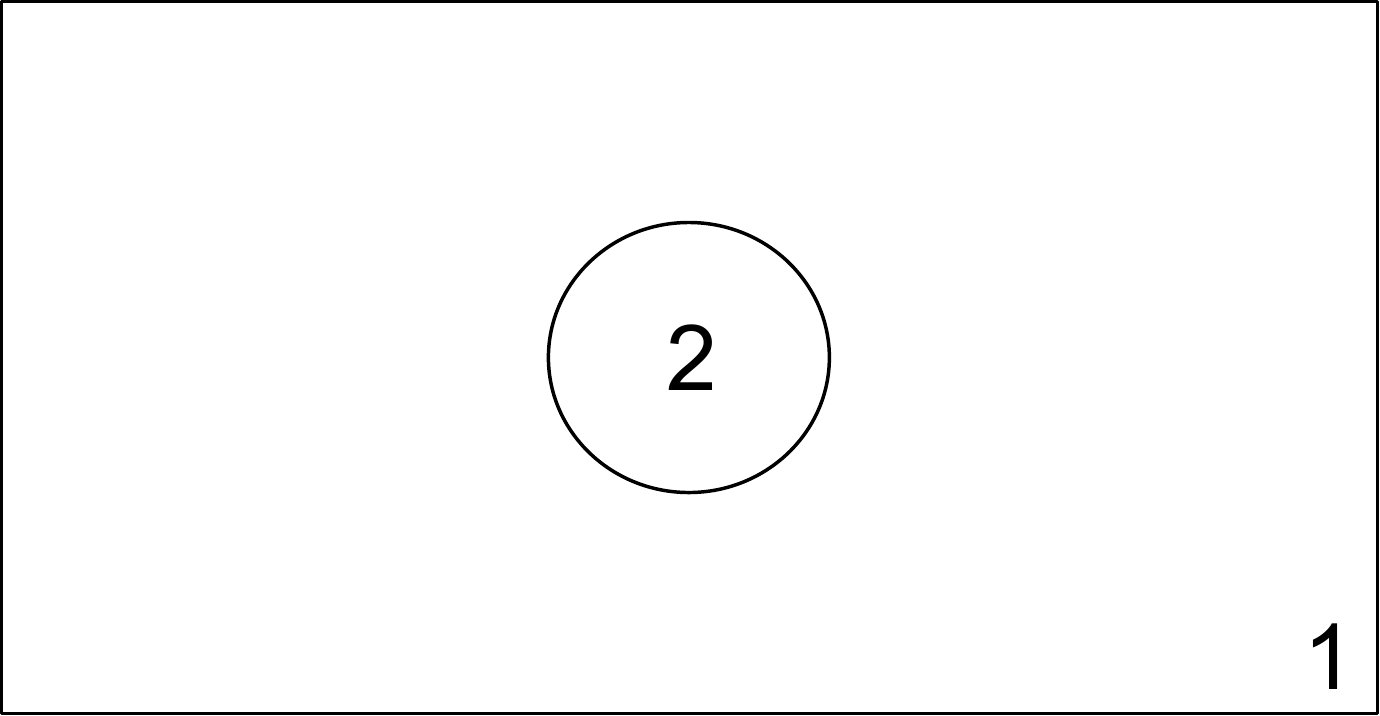}}
  \caption{Canonical ensemble (subsystem 2)  in a heat bath}
  \label{003.fig:CE}
\end{figure}

For subsystem 2 in a heat bath as shown in Fig. (\ref{003.fig:CE}), the total number of cells in this subsystem at fixed energy $E_0$ is approximately given by the following expression
\begin{equation}\label{003.O}
\sum_r\Omega_1(E^0-E_r)\Omega_2(E_r),
\end{equation}
where $\Omega_i(E)$ is the number of cells in subsystem $i$ at fixed energy $E$ and $r$ runs over all the energies subsystem $2$ could have. If the heat bath is large, then $\frac{E_r}{E_0}<<1$ and so we can approximate Eq. (\ref{003.O}) as
\begin{align}
&\sum_r\Omega_1(E^0-E_r)\Omega_2(E_r) \nonumber \\
&\approx \sum_r\Omega_2(E_r)\left(\Omega_1(E^0)+\frac{d\Omega_1}{dE}(E^0-E_r)\right) 
\end{align}
\begin{equation}\label{003.O2}
\approx \sum_r\Omega_2(E_r)\left(Ae^{-\beta E_r}\right)=\sum_s\left(Ae^{-\beta E_s}\right) = AZ,
\end{equation}
where $A$ is some constant, $s$ runs over all the cells in subsystem $2$ and $Z$ is the partition function. 

We can now find the time relative entropy of system in the canonical ensemble with respect to a single celled subsystem by noticing that this subsystem is really contained in a microcanonical ensemble with the single celled subsystem really being a multiple celled subsystem of size $Ae^{- \beta E_s}$. Using the correspondence between our time relative entropy and microcanonical ensemble (with a slight generalization), we find that 
\begin{equation}
S(K_{all}|K_s)=\text{log}\left(\frac{ZA}{Ae^{- \beta E_s}}\right)=\text{log}(Z)+\beta E_s.
\end{equation} 
After averaging over all cells in subsystem 2, we obtain 
\begin{equation}
<S(K_{all}|K_s)>=\text{log}(Z)+\beta<E>,
\end{equation} 
where $<\ldots>$ denotes the average over all cells. This corresponds to the usual expression for entropy in the canonical ensemble. Similar results hold for the grand canonical ensemble. To generalize the result from the microcanonical case, we define the total entropy of the system using the time relative entropy as
\begin{equation}
S=<S(K_{all}|K_s)>.
\end{equation}
Additionally we find that the time relative entropy with respect to subsystems with different volumes is approximately,
\begin{equation}
S(K_{V}|K_{V'})=\text{log}(Z(V))-\text{log}(Z(V')),
\end{equation} 
where $Z(V)$ is the partition function of a subsystem with volume $V$.
\section{Application to Black holes}

We do a heuristic consistency check to see if one can use the time relative entropy to find the entropy of a black hole. To do this, we compare the entropy of two black holes of similar size.

\begin{figure}\centering
  \resizebox{0.40\textwidth}{!}{\includegraphics{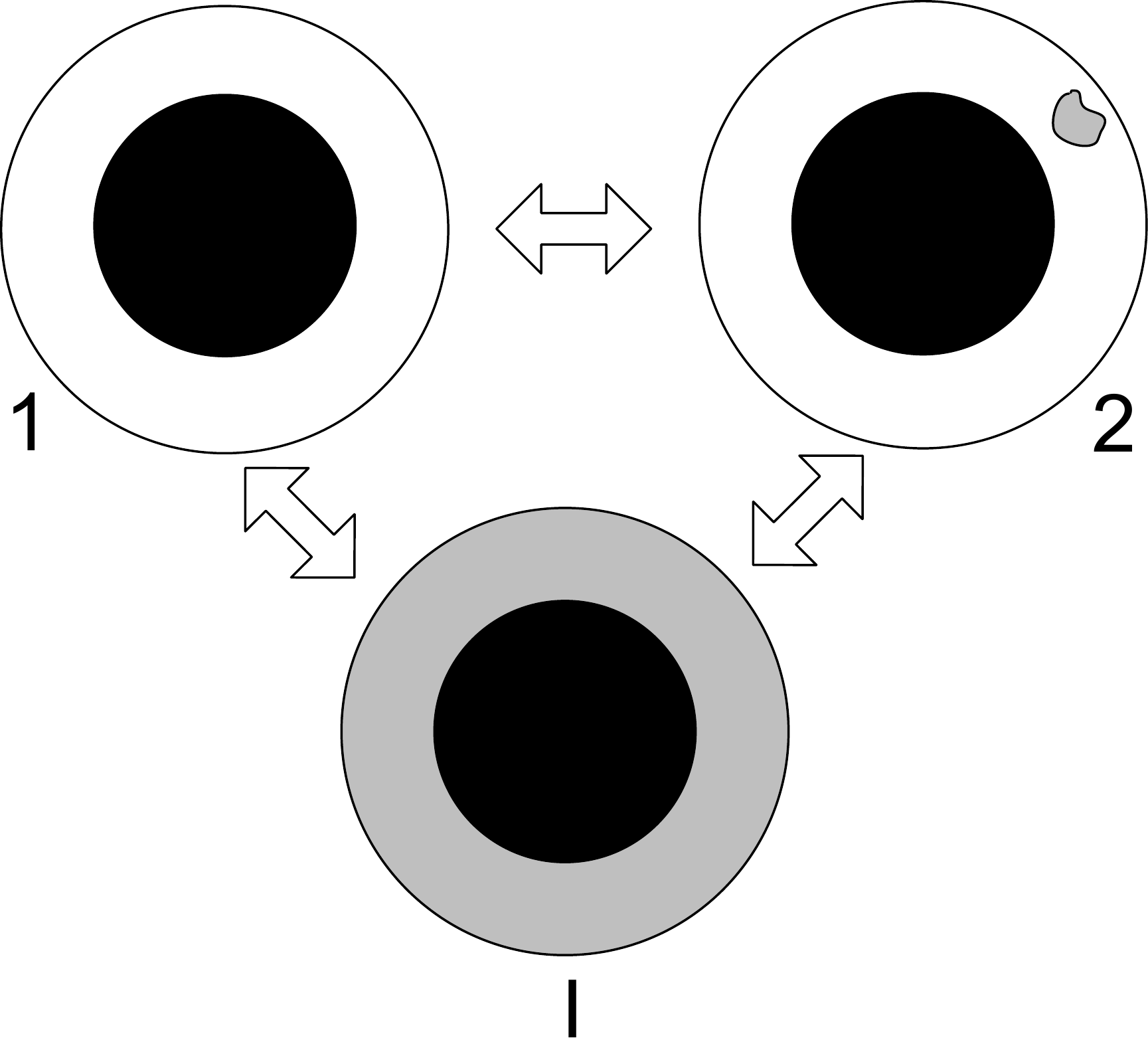}}
  \caption{Black hole ``subsystems" with photon radiation in grey}
   \label{003.fig:BH}
\end{figure}

In Fig. (\ref{003.fig:BH}), the system is a black hole enclosed in a reflective cavity just larger than the black hole, so that any radiation that escapes from the black hole does not escape to infinity. $K_1$ is the subsystem of the non-evaporated black hole and $K_I$ is the subsystem of the evaporated black hole with the radiation at temperature equal to the Hawking radiation.  $K_2$ is the subsystem of the evaporated black but with the radiation in a fixed configuration, so that it is a single cell subsystem ($K_2$ is a subset of $K_I$). We define a small change in the black hole energy as
\begin{equation}\label{003.BH2}
\delta S_{BH}=S(K_1|K_2).
\end{equation}
\begin{remark}
When all the radiation is within a Planck length away from the black hole, we will treat the radiation as being part of the non-evaporated black hole subsystem. In other words, $K_1$ can be view as a subset of $K_I$ and we take our system to be a grand canonical ensemble contained between the reflective cavity and the black hole.   
\end{remark}

\begin{remark}
To be precise, since we need all the subsets to be disjoint, we need to remove $K_1$ and $K_2$ from $K_I$.
\end{remark}

To estimate the ratios of the expected times, note that since $K_I$ is much larger $K_1$ and $K_2$ the system will almost certainly evolve from $K_1$ to $K_2$ via $K_I$, so we can estimate Eq. (\ref{003.BH2}) as

$$
S(K_1|K_2)\approx \text{log}\left(\frac{\tau(K_1 \rightarrow K_I)+\tau(K_I \rightarrow K_2)}{\tau(K_2 \rightarrow K_I)+\tau(K_I \rightarrow K_1)}\right)\approx
$$

$$\text{log}\left(\frac{\tau(K_I \rightarrow K_2)}{\tau(K_I \rightarrow K_1)}\right)\approx \text{log}\left(\frac{\tau(K_I \rightarrow K_2)\tau(K_1 \rightarrow K_I)}{\tau(K_I \rightarrow K_1)\tau(K_2 \rightarrow K_I)}\right)=$$

\begin{equation}
	-\text{log}\left(\frac{\tau(K_I\rightarrow K_1)}{\tau(K_1\rightarrow K_I)}\right)+\text{log}\left(\frac{\tau(K_I\rightarrow K_2)}{\tau(K_2\rightarrow K_I)}\right),
	\end{equation}
where we made the following approximations\\

$\tau(K_I\rightarrow K_i) >> \tau(K_i\rightarrow K_I)$ for $i=1,2$

$\frac{\tau(K_1\rightarrow K_I)}{\tau(K_2\rightarrow K_I)}$ is order of unity.\\

Based on correspondence between the time relative entropy and grand canonical ensemble, we find that 

\begin{equation}\label{003.40}
	\text{log}\left(\frac{\tau(K_I\rightarrow K_2)}{\tau(K_2\rightarrow K_I)}\right)=S_{photons}
\end{equation}

\begin{equation}
	\text{log}\left(\frac{\tau(K_I\rightarrow K_1)}{\tau(K_1\rightarrow K_I)}\right)=\text{log}(Z_V)-\text{log}(Z_V')\approx \frac{1}{4}S_{photons}
\end{equation}

Note: $Z_V$ is the grand canonical partition of radiation in the cavity and $Z_V'$ is the grand canonical partition of the radiation with volume within one Planck length away from black hole. We used that $\text{log}(Z_V)>>\text{log}(Z_V')$ and that $\text{log}(Z_V)=\frac{1}{4}S_{photons}$, derived from the known properties of $Z_V$ for photons .

\begin{remark}
 We have implicitly averaged over all cells on the left-hand side of Eq. (\ref{003.40}) (see section \ref{003.CE}) .
\end{remark}
Hence 
\begin{equation}
\delta S_{BH}\approx \frac{3}{4}S_{photons}=\frac{U}{T}=\frac{\delta M}{T},
\end{equation}
where $U$ is the internal energy of the photon radiation and $M$ is the mass of the black hole. The internal energy of the photon radiation is the mass of the black hole lost in evaporation.

 By integrating this expression and using the expression for $T$ \cite{10.1007/BF02345020} (temperature of Hawking radiation), we find  in natural units that
\begin{equation}
S_{BH}=\frac{A}{4},
\end{equation}
where we have set the arbitrary constant to 0.

The above analysis is reminiscent of calculations of black hole entropy using its thermal atmosphere \cite{tHooft:1984kcu}.

\section{Summary}

Motivated by some of the interpretational issues with black hole entropy, we define a new way to calculate the entropy of a system by determining how irreversible a process is relative to another process rather than counting microstates. 

This new entropy, the time relative entropy, gives a reasonable value for various situations and can be used at least to first approximation to recover the usual entropy of systems in the microcanonical, canonical and grand canonical ensemble. We used the time relative entropy to recover the entropy of a black hole, although the relationship between the mass and temperature of a black hole was needed. 

It would be interesting to investigate Eq. (\ref{003.RE2}) for different functions $f$ with computer simulations or to calculate the ratio of the expected times for simple systems. 

As an extension, it would be interesting to see if the time reversible entropy can be generalized from semi-classical to fully quantum systems.

\bibliographystyle{unsrt}
\bibliography{references}

\begin{thebibliography}{1}

\bibitem{10.2307/79210}
P.~C.~W. Davies.
\newblock The thermodynamic theory of black holes.
\newblock {\em Proceedings of the Royal Society of London. Series A,
  Mathematical and Physical Sciences}, 353(1675):499--521, 1977.

\bibitem{10.1007/BF01645742}
Carter~B. Bardeen, J.M. and Hawking.
\newblock The four laws of black hole mechanics.
\newblock {\em S.W. Commun.Math. Phys.}, 31:161, 1973.

\bibitem{PhysRevLett.26.1344}
S.~W. Hawking.
\newblock Gravitational radiation from colliding black holes.
\newblock {\em Phys. Rev. Lett.}, 26:1344--1346, May 1971.

\bibitem{PhysRevD.7.2333}
Jacob~D. Bekenstein.
\newblock Black holes and entropy.
\newblock {\em Phys. Rev. D}, 7:2333--2346, Apr 1973.

\bibitem{PhysRevD.85.104049}
Aron~C. Wall.
\newblock Proof of the generalized second law for rapidly changing fields and
  arbitrary horizon slices.
\newblock {\em Phys. Rev. D}, 85:104049, May 2012.

\bibitem{PhysRevD.82.124019}
Aron~C. Wall.
\newblock Proof of the generalized second law for rapidly evolving rindler
  horizons.
\newblock {\em Phys. Rev. D}, 82:124019, Dec 2010.

\bibitem{10.1007/BF02345020}
S.W. Hawking.
\newblock Particle creation by black holes.
\newblock {\em Commun.Math. Phys.}, 43, 1975.

\bibitem{tHooft:1984kcu}
Gerard 't~Hooft.
\newblock {On the Quantum Structure of a Black Hole}.
\newblock {\em Nucl.\ Phys.\ B}, 256:727--745, 1985.

\end{thebibliography}

\end{document}